\documentclass[submission,copyright,creativecommons]{eptcs}
\providecommand{\event}{GandALF 2025} 

\usepackage{iftex}
\usepackage{stfloats}
\usepackage{float}
\usepackage{amsthm,amsmath,amssymb,mathtools}
\usepackage{xcolor}
\usepackage{algorithm}
\usepackage[noend]{algpseudocode}
\algrenewcommand{\algorithmiccomment}[1]{\hfill\textcolor{blue}{\(\triangleright\) #1}}
\usepackage{algorithmicx}
\usepackage[noend]{algpseudocode}
\usepackage{booktabs}
\usepackage{tikz}
\usepackage{nicefrac}
\usepackage{xparse}
\usepackage{xfrac}
\usepackage{multirow}
\usepackage{subcaption}
\usepackage{svg}
\usepackage{booktabs}
\PassOptionsToPackage{pdfstartview=FitH, pdfpagemode=UseNone}{hyperref}
\usepackage{hyperref}
\usepackage{xspace}
\usepackage{xcolor}
\usepackage{knowledge}
\usepackage{tikz}\usetikzlibrary{automata, positioning, arrows}
\tikzset {
	->,
	>=stealth',
	node distance=3cm,
	every state/.style={thick, fill=gray!10},
	initial text=$ $,
}

\usetikzlibrary{arrows,positioning,shapes,angles,arrows.meta,shapes.multipart,colorbrewer,backgrounds}
\usepackage{forest}

\definecolor{dartmouthgreen}{rgb}{0.05, 0.5, 0.06}
\definecolor{firebrick}{rgb}{0.7, 0.13, 0.13}

\definecolor{C-3-1}{RGB}{27,158,119}
\definecolor{C-3-2}{RGB}{217,95,2}
\definecolor{C-3-3}{RGB}{117,112,179}

\colorlet{disabled}{lightgray}
\colorlet{statecolor}{black}
\colorlet{rewardcolor}{C-3-2}
\colorlet{actioncolor}{black}
\colorlet{probcolor}{black}
\colorlet{highlightcolor}{C-3-1}

\tikzset{
	every picture/.style={
		>=stealth,
	},
	every edge/.append style={
		thick
	},
	highlight/.style={draw=rewardcolor,thick},
	state/.style={minimum width=0.8cm,minimum height=0.8cm,inner sep=1pt,rectangle,rounded corners,statecolor,draw},
	triangle/.style={regular polygon, regular polygon sides=3},
	max vertex/.style={state,path picture={\draw[gray,rounded corners=0,fill,fill opacity=0.2,thick] (path picture bounding box.south west) -- (path picture bounding box.north) -- (path picture bounding box.south east) -- cycle;}},
	min vertex/.style={state,path picture={\draw[gray,rounded corners=0,fill,fill opacity=0.2,thick] (path picture bounding box.north west) -- (path picture bounding box.south) -- (path picture bounding box.north east) -- cycle;}},
	goal/.style={state,font=\textsf{f},dartmouthgreen,fill=dartmouthgreen!15},
	sink/.style={state,font=\textsf{z},firebrick,fill=firebrick!15},
	action/.style={font=\small,inner sep=0pt,outer sep=3pt},
	reward/.style={inner sep=0pt,outer sep=3pt},
	actionnodelarge/.style={draw=actioncolor,rectangle,rounded corners=2pt},
	actionnode/.style={circle,draw=black,fill=black,minimum size=1mm,inner sep=0,outer sep=0,color=actioncolor},
	actionedge/.style={-,draw,color=actioncolor},
	prob/.style={font=\scriptsize,inner sep=0pt,outer sep=2pt,color=probcolor},
	probedge/.style={->,draw,color=probcolor},
	directedge/.style={->,draw,color=actioncolor}
}

\renewcommand{\circ}{\bigcirc}

\newcommand{\Naturals}{\mathbb{N}}
\newcommand{\Reals}{\mathbb{R}}
\newcommand{\Rationals}{\mathbb{Q}}
\newcommand{\abs}[1]{\lvert #1 \rvert}

\DeclareMathOperator*{\argmax}{arg\, max}

\newcommand{\eqdef}{\vcentcolon=}

\newcommand{\distribution}{\delta}
\newcommand{\Dist}{\mathsf{Dist}}

\newcommand{\post}{\mathsf{Post}}

\newcommand{\para}[1]{\emph{#1}\hspace{0.0em}}

\renewcommand{\Box}{\vartriangle}
\renewcommand{\circ}{\triangledown}
\newcommand{\G}{\mathcal{G}}

\newcommand{\states}{\mathsf{S}}
\newcommand{\minStates}{\states_{\circ}}
\newcommand{\maxStates}{\states_{\Box}}

\newcommand{\act}{\mathsf{A}}
\newcommand{\Av}{\mathsf{Av}}
\newcommand{\trans}{\delta}

\newcommand{\sink}{\mathsf{z}}
\newcommand{\sinks}{\mathsf{Z}}
\newcommand{\target}{\mathsf{f}}
\newcommand{\targets}{\mathsf{F}}

\newcommand{\unknown}{\states^?}

\newcommand{\arbsigma}{\sigma^{\mathsf{arb}}}
\newcommand{\arbtau}{\tau^{\mathsf{arb}}}

\newcommand{\stratsMax}{\Sigma_\Box}
\newcommand{\stratsMin}{\Sigma_\circ}
\newcommand{\Paths}{\mathsf{Paths}}

\newcommand{\reach}{\lozenge}
\newcommand{\stay}{\square}
\newcommand{\probability}{\mathcal{P}}
\newcommand{\val}{\mathsf{V}}
\newcommand{\lb}{\mathsf{L}}
\newcommand{\ub}{\mathsf{U}}

\newcommand{\svireach}{\newlink{def:reach-stay}{\mathsf{reach}}\xspace}
\newcommand{\svistay}{\newlink{def:reach-stay}{\mathsf{stay}}\xspace}

\newcommand{\glb}{\mathsf{l}}
\newcommand{\gub}{\mathsf{u}}
\newcommand{\mec}{\mathsf{Y}}

\newcommand{\decval}{\mathsf{decval}}
\newcommand{\deltastay}{\Delta_{\text{stay}}}

\newcommand{\bestExit}{\mathsf{bexit}}
\newcommand{\bestExitSet}{\mathfrak{B}}

\hypersetup{
	pdfstartview=FitH,  
	pdfpagemode=UseNone 
}

\definecolor{myblue}{RGB}{0, 102, 204}

\newcommand\newtarget[2]{ \phantomsection\hypertarget{#1}{#2}}

\newcommand\newlink[2]{{\protect\hyperlink{#1}{\color{myblue} #2}}}

\newcommand{\trap}{\newlink{def:trap}{\textsf{trap}}\xspace}

\newcommand{\opt}{\newlink{def:opt}{\textsf{opt}}\xspace}

\newcommand{\drawdummy}{\node[minimum size=0,inner sep=0]}

\usepackage{amsthm}  

\theoremstyle{definition}
\newtheorem{definition}{Definition}[section]  

\theoremstyle{plain}
\newtheorem{theorem}[definition]{Theorem}

\theoremstyle{remark}
\newtheorem{remark}[definition]{Remark}
\newtheorem{example}[definition]{Example}

\usepackage{algorithm}
\usepackage{nccmath}
\usepackage[noend]{algpseudocode}
\usepackage{comment}
\usepackage{amssymb}
\usepackage{tabularx}
\usepackage{pgffor}
\usepackage{tikz}
\usepackage{tikz}\usetikzlibrary{automata, positioning, arrows, graphs, quotes, calc,angles, shapes.multipart, trees, patterns,decorations.pathreplacing,decorations.markings}
\usepackage{placeins}
\usepackage{thm-restate}
\usepackage{csquotes}

\tikzset{node distance=2.5cm, 
	every state/.style={minimum size=0pt, fill=gray!10,rectangle,
		align=center},
	initial text={},
	every picture/.style={>=stealth'},
	brace/.style={decorate,decoration=brace}, semithick}
\usepackage{wrapfig}

\newcommand{\myspace}{\vspace*{0em}}
\newcommand{\myspaceb}{\vspace*{0em}}
\renewcommand{\myspace}{\vspace*{0em}}
\renewcommand{\myspaceb}{\vspace*{0em}}

\ifpdf
\usepackage{underscore}         
\usepackage[T1]{fontenc}        
\else
\usepackage{breakurl}           
\fi

\title{Sound Value Iteration for Simple Stochastic Games\thanks{
This research was supported in part by the DFG project GOPro (427755713), the MUNI Award in Science and Humanities (MUNI/I/1757/2021), the EU Horizon Europe Grant (101171844), the EU’s Horizon 2020 Marie
Sklodowska-Curie grant (101034413), and the ERC Starting Grant DEUCE (101077178).}
}
\author{Muqsit Azeem
	\institute{Technical University of Munich}
	\and
	Jan Kretinsky
	\institute{Masaryk University}
	\and
	Maximilian Weininger
	\institute{Ruhr-University Bochum}
}
\newcommand{\titlerunning}{Sound Value Iteration for Simple Stochastic Games}
\newcommand{\authorrunning}{M. Azeem, J. Kretinsky \& M. Weininger}

\hypersetup{
	bookmarksnumbered,
	pdftitle    = {\titlerunning},
	pdfauthor   = {\authorrunning},
	pdfsubject  = {EPTCS},               
}

\begin{document}
	\maketitle
	
	\begin{abstract}
		Algorithmic analysis of Markov decision processes (MDP) and stochastic games (SG) in practice relies on value-iteration (VI) algorithms. Since basic VI does not provide guarantees on the precision of the result, variants of VI have been proposed that offer such guarantees.
		In particular, sound value iteration (SVI) not only provides precise lower and upper bounds on the result, but also converges faster in the presence of probabilistic cycles.
		Unfortunately, it is neither applicable to SG, nor to MDP with end components.
		In this paper, we extend SVI and cover both cases.
		The technical challenge consists mainly in proper treatment of end components, which require different handling than in the literature.
		Moreover, we provide several optimizations of SVI.
		Finally, we evaluate our prototype implementation experimentally to
		demonstrate its potential on systems with probabilistic cycles.
	\end{abstract}
	
\section{Introduction}

\para{Value iteration (VI)} \cite{bellman} is the practically most used method for reliable analysis of probabilistic systems, in particular Markov decision processes (MDPs) \cite{Puterman} and stochastic games (SGs) \cite{condonComplexity}.
It is used in the state-of-the-art model checkers such as Prism \cite{prism4} and Storm \cite{Storm} as the default method due to its better practical scalability, compared to strategy iteration or linear/quadratic programming \cite{revisedPractitionerGuide,gandalf20}.
The price to pay are issues with precision.
Firstly, while other methods yield precise results in theory (omitting floating-point issues), VI converges to the exact result only in the limit.
Secondly, the precision of the intermediate iterations was until recently an open question.
Given the importance of reliable precision in verification, many recent works focused on modifying VI so that the imprecision can be bounded, yielding a stopping criterion.
Consequently, (i) the computed result is reliable, and (ii) the procedure can even terminate earlier whenever the desired precision is achieved.

\begin{wrapfigure}{r}{0.28\textwidth}
	\begin{center}
		\resizebox{0.28\textwidth}{!}{
		\begin{tikzpicture}
			\node[max vertex, ] at (2,0) (s) {$s$};
			\node[actionnode] at (3,0) (sc) {};
			\node[goal] at (4.5,1) (goal) {};
			\node[sink] at (4.5,-1) (sink) {};
			
			\draw[->, thick] (1,0) -- (s);
			
			\path[->]
			;
			\path[actionedge]
			(s) edge node[action,above] {} (sc)
			;
			\path[probedge]
			(sc) edge[out=90,in=75] node[prob,above] {$p$} (s)
			(sc) edge node[prob,above] {$r$} (goal)
			(sc) edge node[prob,below] {{$1-(r+p)~~~~~~~~~$}} (sink)
			;
			\draw[->]  (goal) to[loop right]  node [midway,anchor=west] {} (goal);
			\draw[->]  (sink) to[loop right]  node [midway,anchor=west] {} (sink);
		\end{tikzpicture}
	}
	\end{center}
	\caption{A Markov chain with an initial state $s$, a target state $\target$, and a sink/zero state $\sink$}
	\label{fig:3states}
	
\end{wrapfigure}
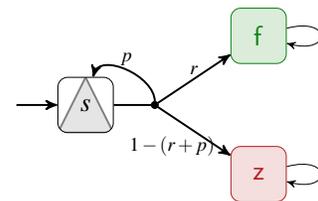

Focusing on reachability properties here, the methods for computing lower and upper bounds on such the reachability probability include `collapsing' \emph{end-components (ECs)} in graphs (\emph{bounded value iteration (BVI)} \cite{atva} a.k.a.\ \emph{interval iteration} \cite{hm18}, \emph{bounded real-time dynamic programming} \cite{atva}), reducing the upper bound in ECs through `deflating'~\cite{KKKW18,DBLP:conf/lics/KretinskyMW23}, finding augmenting-like paths (\emph{widest path} \cite{wp}), guessing the upper bound from the lower one (\emph{optimistic VI} \cite{OVI}), or \emph{sound VI} (SVI) \cite{DBLP:conf/cav/QuatmannK18}.

The main idea of SVI is to estimate the value by a geometric series describing the probabilities to (i)~reach the target in a given number of steps and (ii)~to remain undetermined in that number of steps.
For instance, in the system depicted in Fig.~\ref{fig:3states}, SVI essentially computes that within $n$ steps, the probability to reach the target is $\svireach_n=r+pr+p^2r+\cdots+p^{n-1}r$.
Moreover, since the probability to stay in undetermined states (neither target, nor sink) is $\svistay_n=p^n$, it deduces that the infinite-horizon reachability probability is in the interval $[\svireach_n, \svireach_n+\svistay_n]$.
Not only does the method provide such valid lower and upper bounds, but it also terminates (bounds are close enough to each other) faster in the presence of such probabilistic cycles, in contrast to classical BVI~\cite{DBLP:conf/cav/QuatmannK18}.

Unfortunately, SVI is only applicable to MDP without \emph{end components} (ECs).
This simplifying restriction is common (stochastic shortest path in MDP \cite{ssp}, stopping games \cite{shapley}), but also very severe.
In this paper, \textbf{we extend SVI (i) to SG and (ii) to systems (both MDP and SG) with end components}.

First, we extend SVI to SG without ECs. 
While the algorithm is similar with analogously dual steps for the other player, the correctness is surprisingly not immediate. 
Indeed, the underlying optimization objective in SVI is not simple reachability but a non-trivial combination of step-bounded reachability of targets and of undetermined states.
This in turns makes these strategies require memory; the interaction of two memory-dependent strategies then convolutes constraints on the upper bound.

Second, we extend SVI to handle ECs.
Typically, the literature on guaranteed-precision VI, i.e.~with lower and upper bounds, considers only the case without ECs \cite{DBLP:conf/cav/QuatmannK18}; for MDP, the ECs can be \emph{collapsed} before the analysis, avoiding the issue, but \emph{not so for SG}.
There, the main approach to treat ECs is \emph{deflating} \cite{KKKW18}, which adjusts the upper bound of each state to reflect the true probability of reaching the target \emph{outside} of the EC.
Since SVI does not work with the explicit upper bounds, but rather with an implicit version given by other values (such as $\svireach_n$ and $\svistay_n$), deflating cannot be performed directly.
Consequently, (i)~instead of a single (S)VI update, there are four different cases to handle; (ii)~the fundamental notion of the \emph{simple end component} (~\cite{KKKW18}, EC states with the same value), instrumental for other VI extensions, is too coarse; instead, we have to consider yet smaller parts; and
(iii)~the convergence of the bounds does not follow directly from the literature.
In particular, while convergence of the lower bound uses a simple established proof in all other VI variants, we need properties of the upper bound to argue about the convergence of the lower bound.
This constitutes the main technical challenge.

\para{Our contribution} can be summarized as follows: 
\begin{itemize}
	\item After providing an instructive description of SVI (in particular adding some aspects lacking in the original~\cite{DBLP:conf/cav/QuatmannK18}),
	we extend SVI to SG (Sec.~\ref{sec:svi-ec-free}) and then to handle ECs (Sec.~\ref{sec:algorithm-with-ec}).
	\item Besides, (i) we apply the classical topological optimization to our approach, and we further improve it based on the performed in-depth analysis, as detailed in Sec.~\ref{sec:topological}, including its motivation and integration; (ii) our~prototypical implementation and experimental analysis illustrate that our extension preserves the main advantage of SVI (see Sec.~\ref{svi-good-example} for a concrete example)---fewer iterations are needed for systems with probabilistic cycles.
\end{itemize}
We emphasize that our contribution does \textbf{not} lie in providing a tool more efficient than other ones in the literature, but rather in the theoretically involved \textbf{proof and way how} SVI can be extended to handle ECs.
Specifically, our contribution lies in highlighting the inductive, attractor-like structure of ECs (Sec.~\ref{ssec:bes}), which provides a theoretical foundation for addressing them---the primary challenge with VI. 
\emph{Instead of aiming to replace BVI---which handles ECs (i.e., sure cycles)---this work aims to complement it by offering critical insights into the inherent drawbacks of VI, such as its slow convergence in the presence of probabilistic cycles.	
Our insights pave the way for potential future enhancements to BVI through the efficient probabilistic cycle handling introduced by SVI, which is currently missing in BVI.}
This foundational step, though part of a longer-term effort, is pivotal for enabling subsequent advancements.

\section{Preliminaries}
\label{sec:prelims}

We briefly recall stochastic games and introduce our notation; see e.g.~\cite[Chap. 2]{MaxiThesis} for more details.

\begin{definition}[Stochastic Game (SG)~\cite{DBLP:conf/dimacs/Condon90}]
	A (turn-based simple) stochastic game is a tuple $\G = \langle \states, \maxStates, \minStates, \act, \Av, \delta, \targets \rangle$, where $\states = \maxStates \uplus \minStates$ is a finite set of \emph{states} partitioned into \emph{Maximizer} ($\maxStates$) and \emph{Minimizer} ($\minStates$) states, 
	$\act$ is a finite set of \emph{actions}, the function $\Av: \states \to 2^\act \setminus \{\emptyset\} $ assigns \emph{available} actions to each state, $\delta:\states \times \act \to \Dist(\states)$ is the \emph{transition function} mapping state-action pairs to probability distributions over $\states$, and $\targets\subseteq \states$ is the \emph{target} set.
\end{definition}

An SG with $\minStates=\emptyset$ or $\maxStates=\emptyset$ is called a \emph{Markov decision process (MDP)}. 
For $s \in \states,a \in \Av(s)$, the set of \emph{successors} is $\post(s,a) := \{s' \ | \ \distribution(s,a,s') >0 \} $.
For a set of states $T \subseteq \states$, we use $T_\Box = T \cap \maxStates$ to denote all Maximizer states in $T$, and similarly $T_\circ$ for Minimizer.

The semantics are defined as usual by means of paths and strategies:
An infinite path is a sequence $s_0a_0s_1a_1\ldots \in (\states \times \act)^\omega$ where for each $i \in \mathbb{N}_0$ we have $a_i \in \Av(s_i)$ and $s_{i+1} \in \post(s_i, a_i)$. A history (also called finite path) is a finite prefix of such an infinite path, i.e.\ an element of $(\states\times\act)^* \times \states$ that is consistent with the SG. 
For a history $\rho = s_0a_0s_1a_1 \ldots s_k$, let $\rho_i\vcentcolon=s_i$ denote the $i$-th state in a path and $\abs{\rho} = k$ denote the length of the history (only counting actions; intuitively: the number of steps taken).

A strategy is a (partial) function $\pi \colon (\states\times\act)^* \times \states ~\to~ \act$ mapping histories to actions. 
A Maximizer strategy is only defined for histories that end in a Maximizer state $s\in\maxStates$, and dually for Minimizer.
We denote by $\stratsMax$ and $\stratsMin$ the sets of Maximizer and Minimizer strategies and use the convention that $\sigma\in\stratsMax$ and $\tau\in\stratsMin$.
Complementing an SG $\G$ with a pair of strategies (essentially resolving all non-determinism) induces a Markov chain $\G^{\sigma,\tau}$ with the standard probability space $\probability_{s,\G}^{\sigma,\tau}$  in state $s$, see~\cite[Sec. 2.2]{MaxiThesis}.

The \emph{reachability objective $\reach\targets$} is the set of all infinite paths containing a state in $\targets$.%
\footnote{Formally, denoting the set of all infinite paths by  $\Paths$, $\reach\targets = \{\rho\in\Paths \mid \exists i\in\Naturals. \rho_i \in \targets\}$. We also use the step-bounded variant $\reach^{\leq k}\targets = \{\rho\in\Paths \mid \exists i\in\Naturals. \rho_i \in \targets \wedge i\leq k\}$ and its dual staying objective $\stay^{\leq k}\unknown = \Paths \setminus \Big( \reach^{\leq k}(\states\setminus\unknown) \Big)$.\label{footnote:formalReachEtc}}
The \emph{value} of a state $s$ is the probability that from $s$ the target is reached under an optimal play by both players: $	\val(s) = \sup_{\sigma\in\stratsMax}\inf_{\tau\in\stratsMin} \probability_s^{\sigma,\tau} (\reach \targets) = \inf_{\tau\in\stratsMin}\sup_{\sigma\in\stratsMax} \probability_s^{\sigma,\tau} (\reach \targets)$; the order of choosing the strategies does not matter~\cite{condonComplexity}.
The function $\val: \states \to \mathbb{R}$ maps every state to its value. 
Functions $f_1, f_2: \states \to \mathbb{R}$ are compared point-wise, i.e.\ $f_1 \leq f_2$ if for all $s \in \states$, $f_1(s) \leq f_2(s)$.
In this paper, we denote $f(s, a) = \sum_{s' \in \states} \distribution(s,a,s') \cdot f(s')$ for any $f : \states \rightarrow \mathbb{R}$, using notation overloading.

\subsection{Value iteration and bounded value iteration}
\label{sec:prelims:vi}
Before we define value iteration, we recall a useful partitioning of the state space~\cite{DBLP:conf/cav/QuatmannK18}:
The set of target states $\targets$ is known; w.l.o.g., every target state is absorbing.
Similarly, simple graph search can identify the set $\sinks$ of \emph{sink states} from with no path to $\targets$.
The values of states in these sets are trivially $1$ or $0$, respectively.
Our algorithm focuses on computing the values of the remaining states in $\unknown = \states \setminus (\targets \cup \sinks)$.

Classical value iteration (VI, e.g.~\cite{visurvey}) computes the following sequence $\lb_i:\states\to\mathbb{R}$:

	\begin{align}
		\lb_0(s) &= 
		\begin{cases} 
			1, & \text{if } s \in \targets, \\
			0, & \text{otherwise}.
		\end{cases} \\
		\lb_{i+1}(s) &= 
		\begin{cases}
			\displaystyle \max_{a \in \Av(s)} \sum_{s' \in \states} \distribution(s, a, s') \cdot \lb_i(s'), & \text{if } s \in \maxStates \\
			\displaystyle \min_{a \in \Av(s)} \sum_{s' \in \states} \distribution(s, a, s') \cdot \lb_i(s'), & \text{if } s \in \minStates
		\end{cases} \label{eqn:vi-updates}
	\end{align}

These under-approximations converge (in the limit) to the value $\val$, which also is the \emph{least} fixpoint of Eq.~\eqref{eqn:vi-updates}.
While this approach is often fast in practice, it has the drawback that the imprecision $\val-\lb_i$ is unknown.
To address this, \emph{bounded value iteration} (BVI, also known as interval iteration~\cite{atva,hm18,KKKW18}) additionally computes an over-approximating sequence $\ub_i$ which allows to estimate the imprecision as $\ub_i(s)-\lb_i(s)$. 
Naively, the sequence is initialized as $\ub_0(s) = 0$ for all $s \in \sinks$ and $\ub_0(s) = 1$ for all other states, and then updated using Eq.~\eqref{eqn:vi-updates}.
However, then $\ub_i$ converges to the \emph{greatest} fixpoint of Eq.~\eqref{eqn:vi-updates} which can differ from the least fixpoint $\val$; this happens if and only if \emph{end components} are present~\cite{KKKW18}.
\begin{definition}[End component (EC)~\cite{dA97a}]
	\label{def:EC}
	A set of states $T$ with $\emptyset \neq T \subseteq \states$ is an \emph{end component} if there exists a set of actions $\emptyset \neq B \subseteq \bigcup_{s \in T}Av(s)$ such that: 
	1. for each $s \in T$, $a \in B \cap \Av(s)$ we have $\post(s,a) \subseteq T$, 
	2. for each $s, s' \in T$ there exists a path from $s$ to $s'$ using only actions in $B$.	
	An end component $T$ is a \emph{maximal end component} (MEC) if there is no other EC $T'$ such that $T \subset T'$. Further, a MEC $T$ is \emph{bottom} if for all $s \in T, a\in\act, s' \in \states \setminus T$ we have $\distribution(s,a,s')=0$.

\end{definition}

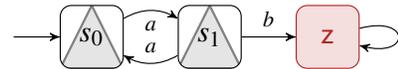
\begin{wrapfigure}{r}{0.35\textwidth}
	\begin{center}
		\resizebox{0.35\textwidth}{!}{
			\begin{tikzpicture}
				\drawdummy (init) at (0,0) {};
				\node[max vertex, ] (q) at (1,0) {$s_0$};
				\node[max vertex] (r) at (2.5,0) {$s_1$};
				\node[sink] (0) at (4,0)  {};
				\draw (init) to (q);
				\draw (q) to [bend left] node[below]{\scriptsize $a$} (r);
				\draw (r) to [bend left] node[above]{\scriptsize $a$} (q);
				\draw (r) to node[above]{\scriptsize $b$} (0);
				\draw (0) to[loop right]  node [midway,anchor=west] {} (0);
			\end{tikzpicture}
		}
	\end{center}
\caption{An SG with an EC $\{s_0,s_1\}$ and a sink $\sink$.}
\label{fig:2states}
\end{wrapfigure}

As an example, consider the MDP in Fig.~\ref{fig:2states}. Here, any $\ub$ with $\ub(s_0)=\ub(s_1)$ and $\ub(\sink)=0$ is a fixpoint of Eq.~\eqref{eqn:vi-updates}.
More generally, on an EC in a maximizing MDP, one can create a fixpoint by assigning any number larger than the actual value to all states of the EC (and dually, a smaller number in a minimizing MDP).
In MDP, there is a rather simple solution: The entire EC can be  collapsed into a single state, keeping only leaving actions. This ensures a single fixpoint and, since all states in an EC have the same value as the collapsed state, it allows to obtain the values for the original MDP~\cite{atva,hm18}.
For SG, this solution is not applicable because states in an EC can have different values~\cite[Sec. 3]{KKKW18}.
The solution of~\cite{KKKW18} is to repeatedly identify subsets of ECs that possibly, given the current estimates, have the property that all states have the same value; then, the over-approximation in these subsets is conservatively lowered (\emph{``deflated''}) to the over-approximation of the best action \emph{exiting} the subset.
Relevant for our work is the fact that we need to decrease the over-approximation inside ECs to the estimate of the \emph{best exit}. Thus, we repeat here a (slightly adjusted) definition of best exit. 
\begin{definition}[Best exit~\cite{KKKW18}]
	\label{def:bestExit}
	Given a set of states $T \subseteq \states$ and a function $f : \states \rightarrow \Rationals$, the best exit
	according to $f$ from $T$ is defined as:
	\begin{align*}
		\bestExit^{\Box}_f(T) &= 
		\displaystyle \argmax_{\substack{(s,a) \in T_\Box\times \Av(s)\\ \post(s,a) \nsubseteq T}} \ \sum_{s' \in \states} \distribution(s,a,s') \cdot f(s') 
	\end{align*}
\end{definition}

\subsection{Sound value iteration in Markov chains}

\label{sec:prelims:svi}

Sound value iteration (SVI,~\cite{DBLP:conf/cav/QuatmannK18}) is an alternative to BVI that often is faster while still providing guarantees on the precision. 
It computes for increasing $k$ the probability of reaching the target set $\targets$ within $k$ steps, denoted $\probability(\reach^{\leq k} \targets)$, and the probability to stay within the unknown set $\unknown$ during the first $k$ steps, denoted $\probability(\stay^{\leq k} \unknown)$ --- see Footnote~\ref{footnote:formalReachEtc} for the formal definitions of the operators $\reach^{\leq k}$ and $\stay^{\leq k}$.
From these two sequences, it computes under- and over-approximations of the value utilizing three core ideas.
\begin{itemize}
	\item The value can be split into the probability to reach the target set \emph{within $k$ steps} plus the probability to reach the target set \emph{in strictly more than $k$ steps}~\cite[Lem. 1]{DBLP:conf/cav/QuatmannK18}.
	The former $\probability(\reach^{\leq k} \targets)$ is directly computed by the algorithm; the latter is approximated using the next two ideas.
	
	\item A path that reaches the target set in strictly more than $k$ steps first stays in $\unknown$ for $k$ steps and afterwards reaches the target. As we do not know which state in $\unknown$ is reached after the $k$ steps, we estimate its value using a lower bound $\glb_k$ or an upper bound $\gub_k$ on the value over \emph{all states in $\unknown$}.
	Formally, the probability to reach in strictly more than $k$ steps can be under-approximated as $\probability(\stay^{\leq k} \unknown)\cdot\glb_k$; dually, it can be over-approximated as $\probability(\stay^{\leq k} \unknown)\cdot\gub_k$~\cite[Prop.~1]{DBLP:conf/cav/QuatmannK18}.
	
	\item Such bounds $\glb_k$ and $\gub_k$ can be obtained from the $k$-step bounded probabilities $\probability(\reach^{\leq k} \targets)$ and  $\probability(\stay^{\leq k} \unknown)$~\cite[Prop. 2]{DBLP:conf/cav/QuatmannK18}.
	Intuitively, the highest possible value would be obtained if the state $\hat{s}\in\unknown$ with the highest value always reached itself after $k$ steps --- getting another try to obtain the highest possible value. The resulting value can be computed as $\sum_{i=0}^\infty \probability_{\hat{s}}(\stay^{\leq k} \unknown)^i\cdot\probability_{\hat{s}}(\reach^{\leq k} \targets)  $. 
	Assuming $\probability_{\hat{s}}(\stay^{\leq k} \unknown) < 1$, we can write this geometric series in closed form:

	\begin{align}
		\sum_{i=0}^\infty \probability_{\hat{s}}(\stay^{\leq k} \unknown)^i\cdot\probability_{\hat{s}}(\reach^{\leq k} \targets) 
		= \frac{\probability_{\hat{s}}(\reach^{\leq k} \targets)}{1-\probability_{\hat{s}}(\stay^{\leq k} \unknown)} \notag
		\leq \max_{s\in\unknown} \frac{\probability_{s}(\reach^{\leq k} \targets)}{1-\probability_{s}(\stay^{\leq k} \unknown)} \notag
		=\colon \gub_k
	\end{align}

	Dually, we can compute a lower bound $\glb_k$ by picking the state in $\unknown$ that minimizes the closed form of the geometric series.
	We exemplify the power choosing lower and upper bounds like this:
	Consider again the Markov chain of Fig.~\ref{fig:3states} and set $p=0.98$ and $r=0.01$. While BVI requires 682 iterations to achieve a precision of $10^{-6}$, SVI takes precisely 1 iteration.
	SVI's behaviour is in fact independent of $p$ and $r$, whereas increasing $p$ also increases the number of iterations of BVI.
\end{itemize}

We restate~\cite[Thm. 1]{DBLP:conf/cav/QuatmannK18}, which summarizes the above insights and proves the approach correct.

\begin{theorem}
	\label{thm:svi-mc} 
	Let $\mathcal{M}$ be a Markov chain with probability measure $\probability$ whose state space is partitioned into $\targets$, $\unknown$ and $\sinks$ as described above. 
	Let $k\geq0$ such that $\probability (\square^{\leq k} \unknown) < 1$ for all $s \in \unknown$. 
	Then, for every state $s\in\unknown$, its value $\val(s)=\probability_s(\reach \targets)$ satisfies the following:

\begin{align*}
\probability_s(\reach^{\leq k} \targets) &+ \probability_s
(\square^{\leq k} \unknown) \cdot \glb_k 
\leq \val(s) \leq 
 \probability_s
(\reach^{\leq k} \targets) + \probability_s (\square^{\leq k} \unknown)
\cdot \gub_k
\end{align*}
	
	where $\glb_k=\min_{s' \in \unknown} \frac{\probability_{s'}(\reach^{\leq k} \targets)}{1-\probability_{s'} (\square^{\leq k} \unknown)}$ and dually $\gub_k=\max_{s' \in \unknown} \frac{\probability_{s'}(\reach^{\leq k} \targets)}{1-\probability_{s'} (\square^{\leq k} \unknown)}$.
	
\end{theorem}

The termination of the algorithm relies on the assumption that the system is contracting towards $\targets\cup\sinks$; intuitively, this means that there are no sure cycles in $\unknown$ and thus a path has to leave $\unknown$ almost surely.
However, in MDP and SG there can be sure cycles, namely the ECs, which also hinder the convergence of BVI.  
In Sec.~\ref{sec:algorithm-with-ec}, we provide a way to deal with the presence of ECs when using SVI; note that neither the previous solution of collapsing~\cite{atva,hm18} nor of deflating~\cite{KKKW18} suffice.

However, before we address ECs, Sec.~\ref{sec:svi-ec-free} turns to a different complication:
Our reasoning so far concerned Markov chains without non-determinism. 
When generalizing to MDP (or SG), the computations additionally depend on the choice of strategy (strategies). In particular, the probability to reach within $k$ steps depends on the chosen strategies and need not be the optimal $k$-step value; moreover, the bounds cannot be computed using only the closed form of the geometric series.

	\section{Sound value iteration for stochastic games without end components}\label{sec:svi-ec-free}

This section extends SVI from MDP to SG without ECs. 
We build on previous work~\cite{DBLP:conf/cav/QuatmannK18}, performing necessary dual operations for adding a player.
Further, this section prepares the reader for the more technical discussion in Sec.~\ref{sec:algorithm-with-ec} as it introduces the basic algorithm, revisits SVI concepts with additional examples, and improves notation.
Also, we complete the original correctness proof, see Rem. \ref{rem:missingInQK18}.

In order to lift Thm.~\ref{thm:svi-mc} from Markov chains to SG, Sec.~\ref{sec:svi-strats} and~\ref{sec:no-ec-bounds} answer the following two questions, respectively:
How do we pick the strategies $\sigma_k$ and $\tau_k$ that Maximizer and Minimizer use in the $k$-th iteration? 
How do we pick the lower bound $\glb_k$ and the upper bound $\gub_k$ on the value of all states in $\unknown$?
Finally, Sec.~\ref{sec:algorithm-no-ec} combines the insights into the full algorithm.

\subsection{Computing the strategies}
\label{sec:svi-strats}
We following the idea of~\cite[Sec. 4.1]{DBLP:conf/cav/QuatmannK18}:
For an optimal Maximizer strategy $\sigma^{\text{max}}$ and for all Minimizer strategies $\arbtau$ and bounds $\gub_k \geq \max_{s\in\unknown} \probability_s^{\sigma^{\max},\arbtau}(\reach \targets)$, the following chain of inequations holds:

\begin{align}
	\val = \sup_{\sigma\in\stratsMax}\inf_{\tau\in\stratsMin} \probability^{\sigma,\tau}(\reach \targets) &=\inf_{\tau\in\stratsMin} \probability^{\sigma^{\text{max}},\tau}(\reach \targets)\notag 
	\leq \probability^{\sigma^{\text{max}},\arbtau}(\reach \targets) 
	\nonumber \\
	&\leq 
	\probability^{\sigma^{\text{max}},\arbtau}(\reach^{\leq k} \targets) + \probability^{\sigma^{\text{max}},\arbtau}(\stay^{\leq k} \unknown)\cdot \gub_k
	\nonumber \\
	&\leq \sup_\sigma \left(\probability^{\sigma,\arbtau}(\reach^{\leq k} \targets) + \probability^{\sigma,
		\arbtau}(\stay^{\leq k} \unknown)\cdot \gub_k\right)
	\label{eq:maxCorr}\hspace*{-4mm}
\end{align}
We can use the dual argument for a lower bound $\glb_k$, an arbitrary Maximizer strategy $\arbsigma$ and an optimal Minimizer strategy $\tau^{\text{min}}$ to obtain:
\begin{align}
	\label{eq:minCorr}
	\val \geq \inf_\tau \left(\probability^{\arbsigma,\tau}(\reach^{\leq k} \targets) + \probability^{\arbsigma,\tau}(\stay^{\leq k} \unknown)\cdot \glb_k\right)
\end{align}

Based on these inequalities, we can inductively construct strategies $\sigma_k,\tau_k$.%
\newtarget{def:opt}
In order to highlight the symmetry of the players, our formal definition uses a strategy $\pi$ of either Maximizer or Minimizer.
There are two differences between the players: the operator $\opt$ is $\max$ if $\pi\in\stratsMax$ and dually $\min$ if $\pi\in\stratsMin$ (with $\overline{\opt}=\min$ if $\opt=\max$ and vice versa) and the bound $\mathsf{b}_k$ is either $\gub_k$ or $\glb_k$, respectively.
Importantly, the strategies do not optimize ($k$-step) reachability, but the $k$-step objective $\probability^{\sigma_k,\tau_k}(\reach^{\leq k} \targets) + \probability^{\sigma_k,\tau_k}(\stay^{\leq k} \unknown)\cdot \mathsf b_k$, as suggested by Eqs.~\eqref{eq:maxCorr} and~\eqref{eq:minCorr}.
We construct the strategies recursively, building on the correctness of the strategies with shorter horizon. To achieve this, the $(k+1)$st strategy makes one optimal choice in the first step and then mimics the $k$th strategy, i.e.\ for a path $sa\rho$ starting with state $s$ and chosen action $a$, afterwards we have $\pi_{k+1}((sa)\rho)=\pi_k(\rho)$.
As base case, when $k=0$ we choose an arbitrary action, because all actions are 0-step optimal.
Formally,

{\small
	\begin{equation}
	\label{eq:strats}
	\boxed{
		\begin{aligned}
			&\textsc{``$k$-step'' Strategies:} \\
			&\pi_0(\rho) \in \Av(\rho_0)\\
			&\pi_{k+1}(\rho) \gets 
		\begin{cases}
				\begin{aligned}
					\displaystyle\in \arg\opt_{a \in \Av(s)}  \sum_{s' \in S} \distribution(s, a, s') \cdot \Big( 
					\probability^{\sigma_k, \tau_k}_{s'}(\reach^{\leq k}\targets) + 
					\probability^{\sigma_k, \tau_k}_{s'}(\stay^{\leq k}\unknown) \cdot \mathsf{b}_k \Big),
				\end{aligned}
				& \mbox{if } \rho = s \in \states \\
				\pi_k(\rho'), & \mbox{if } \rho = sa\rho'
			\end{cases}
		\end{aligned}
	}
	\end{equation}
}



For histories longer than $k+1$, the definition is irrelevant and can be arbitrary as we shall only consider $(k+1)$-step optimality.
These strategies yield correct over-/under-approximations in the sense of Thm.~\ref{thm:svi-sg-ecfree} below, lifting Thm.~\ref{thm:svi-mc}. 
In the extended version~\cite{svi-sg-extended}, we state a simplified version of this result with constant bounds $\glb,\gub$, as opposed to the generally variable bounds of Thm.~\ref{thm:svi-sg-ecfree}

For readability, we omit the specific strategies $\sigma_k,\tau_k$ and the objectives when possible:%
\newtarget{def:reach-stay} 
Assuming the players use $\sigma_k$ and $\tau_k$ as defined in Eq. \eqref{eq:strats}, we denote the probability to reach $\targets$ or stay in $\unknown$ within $k$ steps from $s$ by $\svireach_s^k \vcentcolon= \probability_s^{\sigma_{k},\tau_{k}} (\reach^{\leq k} \targets)$ and $\svistay_s^k \vcentcolon= 	\probability_s^{\sigma_{k}, \tau_{k}} (\square^{\leq k} \unknown)$, respectively.

\myspaceb
\subsection{Computing bounds}
\label{sec:no-ec-bounds}

Having answered the first question---how to compute strategies---we turn to the second question: obtaining good bounds $\glb_k$, $\gub_k$.
While trivial bounds $\glb=0$ and $\gub=1$ suffice for a correct and convergent algorithm, using good bounds is the key to the good practical performance of SVI. 
A natural choice is $\glb_k = \min_{s \in \unknown} \frac{\svireach_s^k}{1-\svistay_s^k}$ (dually for $\gub_k$).
However, in the presence of non-determinism, this is incorrect~\cite[Sec. 4]{DBLP:conf/cav/QuatmannK18}.
Instead, the updates have to proceed more carefully, not overshooting the interval of values for which the currently considered action is relevant.
This safety cap is called the \emph{decision value}.
We refer the reader to the extended version~\cite{svi-sg-extended} for an extensive derivation and explanation.
Apart from dualization and adjusted notation, our definition is as in~\cite{DBLP:conf/cav/QuatmannK18}.
Denote the difference in staying probabilities of actions $\alpha$ and $\beta$ by $\deltastay(\alpha,\beta) \eqdef \sum_{s' \in \states} (\trans(s,\alpha,s') - \trans(s,\beta,s'))\cdot\svistay^k_{s'}$.
\begin{equation*}
	\label{eq:def-dval-sa}
	\boxed
	{
		\begin{aligned}
			&\decval_{\mathsf{b}}^k(s,\alpha)\eqdef  \opt_{\beta\in\Av(s). \deltastay(\alpha,\beta) > 0} ~ \frac
			{\sum_{s' \in \states} (\trans(s,\beta,s') - \trans(s,\alpha,s'))\cdot\svireach^k_{s'}}
			{\sum_{s' \in \states} (\trans(s,\alpha,s') - \trans(s,\beta,s'))\cdot\svistay^k_{s'}}	\\
			&\displaystyle
			\decval_{\mathsf{b}}^k \eqdef \opt (\decval_{\mathsf{b}}^{k-1}, \opt_{s \in \unknown} \decval_{\mathsf{b}}^k(s,\alpha)) 
		\end{aligned}
	}
\end{equation*}
where the second line sets the global decision value as the optimum over all states, reflecting also previous decision values.
Finally, we can pick the best possible bound using the geometric series, but capping it at the decision value. The outer $\overline{\opt}$ ensures monotonicity of the bounds:
\begin{equation}
	\label{eq:def-bounds}
	\boxed{
		\begin{aligned}
			&\textsc{(Dynamic) Bounds:}\\
			&\mathsf b_k \eqdef  \overline{\opt} \left(\mathsf b_{k-1}, \opt~(\opt_{s \in \unknown}~ \frac{\svireach_s^k}{1-\svistay_s^k}, \decval_{\mathsf{b}}^k) \right)
		\end{aligned}
	}
\end{equation}

The condition \textsc{`updateGlobalBounds?'} ensures that we only update the global bounds when all the states in $\unknown$ have $\svistay$ strictly less than 1, as otherwise $\frac{\svireach_s^k}{1-\svistay_s^k}$ would be undefined.
\begin{equation}
	\boxed{
		\begin{aligned}
			\label{eqn:updateGlobalBounds-no-ec}
			\textsc{updateGlobalBounds}?:= \forall s \in \unknown. \svistay^k_s< 1 
		\end{aligned}
	}
\end{equation}

Thm.~\ref{thm:svi-sg-ecfree} states the correctness of the approximations computed based on the \enquote{$k$-step} strategies from Eq.~\eqref{eq:strats} and dynamically changing bounds from Eq.~\eqref{eq:def-bounds}.

\begin{restatable}{theorem}{theoremsvisgecfree}
		\label{thm:svi-sg-ecfree} 
	
	Fix an SG $\G$ with its state space partitioned in $\targets$, $\sinks$ and $\unknown$. Let $\glb_k, \gub_k \in \Reals$ be bounds as defined by Eq.~\eqref{eq:def-bounds}.
	For $\sigma_k \in  \stratsMax, \tau_k \in \stratsMin$ defined by Eq.~\eqref{eq:strats}, 
	the following inequality holds:	
	

	\begin{equation*}\probability_s^{\sigma_k,\tau_k} (\reach^{\leq k} \targets) +
		\probability_s^{\sigma_{k},\tau_{k}} (\square^{\leq k} \unknown) \cdot \glb_k
		\leq
		\val(s)
		\leq
		\probability_s^{\sigma_{k},\tau_{k}} (\reach^{\leq k} \targets) +
		\probability_s^{\sigma_{k}, \tau_{k}} (\square^{\leq k} \unknown) \cdot \gub_k
	\end{equation*}
\end{restatable}

\begin{proof}[Proof sketch]
	First, we fix an arbitrary $k\in\Naturals$ and prove 
    that $\val(s)\leq
    \gub_k$. Second, we instantiate the arbitrary Minimizer strategy with $\tau_k$ and from Ineq.~\eqref{eq:maxCorr} obtain the following:
	\begin{align}
		\val(s) \leq \sup_{\sigma\in\stratsMax} \left( \probability_{s}^{\sigma,\tau_k} (\reach^{\leq k} \targets) +
		\probability_{s}^{\sigma,\tau_k} (\square^{\leq k} \unknown) \cdot \gub_k \right)
	\end{align}
	Finally, we use induction and an argument about the decision value to prove that $\sigma_k$ is among the optimal strategies for the modified objective.
\end{proof}	
	\begin{remark}\label{rem:missingInQK18}
		The final statement ($\sigma_k$ being an optimal strategy) was not proven in the original paper~\cite{DBLP:conf/cav/QuatmannK18}.
		It only stated its necessity and described how to construct $\sigma_k$.
		We provide the (2-page long) proof in the extended version~\cite{svi-sg-extended}, 
		including the induction that proves the optimality of $\sigma_k$ and utilizes the decision value.
	\end{remark} 
\vspace{-0.1in}
\subsection{Algorithm}
\label{sec:algorithm-no-ec}
\vspace{-0.05in}

We summarize the SVI algorithm for SG without ECs in Alg.~\ref{alg:svi-sg-high-level}. We also provide an equivalent,
fully formal description of the algorithm in the extended version~\cite{svi-sg-extended}.
\begin{algorithm*}[h]
	\caption{Sound value iteration for stochastic games without end components}
	\label{alg:svi-sg-high-level}
	\begin{algorithmic}[1]
		\Require SG $\G$ (partitioned into $\unknown,\targets,\sinks$) and precision $\varepsilon>0$
		\Ensure $\val' \colon \states \to \Reals$ such that $\abs{\val'(s)-\val(s)} \leq \varepsilon$ for all $s \in \states$

		\State For $s\in \targets$ and all $k\in\Naturals_0$, $\svireach_s^k \gets 1$ and $\svistay_s^k \gets 0$  \label{line:svi:initStart} \Comment{Intialization}
		\State For $s\in \sinks$ and all $k\in\Naturals_0$, $\svireach_s^k \gets 0$ and $\svistay_s^k \gets 0$
		\State For $s\in \unknown$, $\svireach_s^0 \gets 0$ and $\svistay_s^0 \gets 1$ 
		\State $\glb_0\gets0$ and $\gub_0\gets1$ \label{line:svi:initEnd}
		\State $k\gets 0$
		
		\smallskip
		
		\Repeat
		\State Use \textsc{``$k$-step'' Strategies} to choose $\sigma_{k+1}$ and $\tau_{k+1}$ -- see Eq.~\eqref{eq:strats} \label{line:svi:chooseStrats}
		\For{all $s\in\unknown$}  \Comment{$(k+1)$-step bounded probabilities}
		\State Compute $\svireach_s^{k+1}$ and $\svistay_s^{k+1}$ by \textsc{Bellman-Update} 
		-- see Eq.~\eqref{alg:update-no-ec}
		\label{line:svi:updateStepBoundProbs}
		\EndFor
		\If{\textsc{updateGlobalBounds}? --  see condition~\eqref{eqn:updateGlobalBounds-no-ec}}  \Comment{Update bounds}
		\State Compute $\glb_{k+1}$ and $\gub_{k+1}$ by \textsc{(Dynamic) Bounds} -- see Eq.~\eqref{eq:def-bounds}
		\EndIf
		\label{line:svi:updateGlobalBounds-l-u}
		\State $k \gets k+1$
		\Until{$\svistay^k_s \cdot (\gub_k-\glb_k) < 2\cdot\varepsilon$ for all $s\in\unknown$} \Comment{Termination condition} 
		\label{line:svi:terminationCond}
		\State	\Return $\svireach^k_s + \svistay^k_s \cdot \frac{\glb_k+\gub_k}{2}$ for all $s \in \unknown$
		
	\end{algorithmic}
\end{algorithm*}
The algorithm first initializes the 0-step bounded probabilities and upper and lower bounds using the trivial under- and over-approximations of 0 and 1 (Lines~\ref{line:svi:initStart}-\ref{line:svi:initEnd}).
For targets and sinks, this trivial initialization is done for all step bounds $k$.
Afterwards, the main loop of the algorithm repeats the following steps:
First, as discussed in Sec.~\ref{sec:svi-strats}, Line~\ref{line:svi:chooseStrats} chooses the strategies optimizing the $k$-step objective $\probability^{\sigma_k,\tau_k}(\reach^{\leq k} \targets) + \probability^{\sigma_k,\tau_k}(\stay^{\leq k} \unknown)\cdot \mathsf b_k$. 
This only relies on $k$-step values which have been computed in the previous iteration.
Second, Line~\ref{line:svi:updateStepBoundProbs} computes the $(k+1)$-step bounded probabilities for all the states according to our strategies, again building on the $k$-step values:

\begin{equation}
	\label{alg:update-no-ec}
	\boxed{
		\begin{aligned}
			&\textsc{Bellman-Update:}\\
			&\svireach^{k+1}_s \gets \sum_{s' \in \states}\distribution(s,\pi_{k+1}(s), s') \cdot \svireach^{k}_{s'};~~~~\svistay^{k+1}_{s} \gets \sum_{s' \in \states}\distribution(s,\pi_{k+1}(s), s') \cdot \svistay^{k}_{s'}
		\end{aligned}
	}
\end{equation}

Finally, if \textsc{updateGlobalBounds}? evaluates to true, Line~\ref{line:svi:updateGlobalBounds-l-u} updates the global bounds using Eq.~\eqref{eq:def-bounds}. 
The loop is repeated until the approximations of lower and upper bounds are $\varepsilon$-close (Line~\ref{line:svi:terminationCond}). 
This happens when $\svistay$ comes close to 0 or the global lower and upper bound become close.
One can also use relative difference by replacing $\svistay^k_{s}$ by $\frac{\svistay^k_{s}}{\svireach^k_{s}+\svistay^k_{s} \cdot \gub}$. 
We prove in the extended version~\cite{svi-sg-extended}:

\begin{restatable}{theorem}{theoremsvinoec}
	\label{thm:svi:no-ec}
Fix an SG $\G$ without ECs in $\unknown$, and fix a precision $\varepsilon>0$. 
Alg.~\ref{alg:svi-sg-high-level} terminates returning $\val'$, an $\varepsilon$-approximation of the value, i.e. for all $s \in \unknown, |\val'(s)- \val(s)| \leq \varepsilon$.
\end{restatable}

	\section{Extending sound value iteration to systems with end components}
\label{sec:algorithm-with-ec}
This section first demonstrates that  Alg.~\ref{alg:svi-sg-high-level} does not converge in the presence of ECs. 
Second, we explain that natural extensions of the \emph{deflate} solution for BVI~\cite{KKKW18} do not resolve the problem for SVI.
Third, we introduce new notions of \emph{best-exit set (BES)} and \emph{delay} action, based on which we design a valid solution.
\newpage

\subsection{Approximations do not converge in presence of ECs} 
\label{sec:non-convergence}

\tikzset{node distance=2.5cm, 
	every state/.style={minimum size=15pt, fill=white, circle, 
		align=center, draw}, 
	chance state/.style={minimum size=3pt, inner sep=0pt, fill=black, circle, 
		align=center, draw},
	max state/.style={minimum size=20pt, fill=white, rectangle, 
		align=center, draw},
	every picture/.style={-stealth},
	brace/.style={decorate,decoration=brace}, semithick}

\begin{wrapfigure}{r}{0.35\textwidth}
	\centering
	\vspace*{-3em}
	\resizebox{0.35\textwidth}{!}{
		\begin{tikzpicture}
			\node[max vertex, ] (s) {$s$};
			\node[chance state] at (1,0)(c){};
			\node[goal] at (2.5,0.5) (f) {};
			\node[sink] at (2.5,-0.5) (z) {};
			
			\draw (f) to [out=100,in=60,loop,looseness=5] (f);
			\draw (z) to [out=290,in=250,loop,looseness=5] (z);
			\draw (s) to [out=150,in=210,loop,looseness=5] node[right] {$a$}  (s);
			\path[actionedge]
			(s) edge node[action,above] {$b$} (c);
			\path (c) edge node[above] {\scriptsize{$\nicefrac12$}} (f);
			\path (c) edge node[below] {\scriptsize{$\nicefrac12$}}(z);
		\end{tikzpicture}
	}
	\vspace*{-2em}
	\caption{A game with trivial Maximizer EC}
	\label{fig:example1}
\end{wrapfigure}
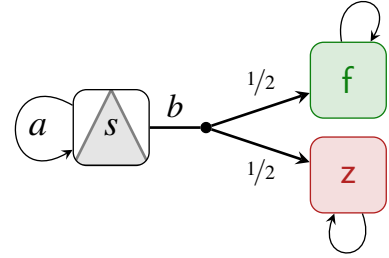
The termination proof of Thm.~\ref{thm:svi:no-ec} crucially relies on the absence of ECs. 
Even the simple EC in Fig.~\ref{fig:example1} leads to non-termination of Alg~\ref{alg:svi-sg-high-level}: 
After initialization, we have $\glb=0$, $\gub=1$, $\svireach^k_{s}=0$, 
and 
$\svistay^k_{s}=1$.
For $k\geq1$, the algorithm picks $\sigma(s):=a$, resulting in no change to the approximations.
In particular, the over-approximation is always $1$ and does not converge to the true $\nicefrac12$.
The key element of any solution to the EC issue is to identify where to best \emph{exit} the EC, which provides realistic information on the value compared to non-informative cyclic dependencies within the EC.

\subsection{A sequence of difficulties: Why the natural solutions do not work}
\para{Root of the difficulties: Incompatibility of SVI and deflating BVI.} ECs were previously handled for \emph{BVI} using so-called \emph{deflating}~\cite{KKKW18}.
This approach finds so-called \emph{simple end components (SEC)}, which are ECs where all states having the same value.
In each SEC, deflating reduces the individual over-approximations for each state in the SEC to the value of Maximizer's \emph{best exit} from the SEC (see Def.~\ref{def:bestExit}), since none of these states can possibly get any higher value. 
In contrast, the difficulty with \emph{SVI} is that it does not maintain explicit local over-approximations for each state which could be reduced (as used in deflating for BVI).
Instead, it computes over-approximations from the $k$-step values $\svireach$, $\svistay$ and a global bound $\gub$.
Consequently, combining SVI with deflating is surprisingly involved.  Since SVI requires expressing the desired probability as a geometric series, we need to cut ECs into even smaller pieces than SECs.

 We outline key difficulties below, with detailed descriptions in 
 ~\cite{svi-sg-extended} in the interest of space.
\begin{enumerate}
	\item[(\textbf{1})] Although SVI maintains an over-approximation $\gub$, it is a \emph{global} bound for all states and thus cannot be used to locally deflate the value of some states.

	\item[(\textbf{2})] SVI derives the over-approximation for each state from both $\svireach$ and $\svistay$.
	However, as it depends on both variables, the same over-approximation can correspond to different pairs of these values. 
	Consequently, it is unclear whether to deflate $\svireach$ or $\svistay$; moreover, there may not even exist any choice that works correctly for the whole EC (or SEC).

	\item[(\textbf{3})] Deflating only the best exit states of a SEC may result in non-convergence of the other states within the SEC.
	
	\item[(\textbf{4})] Both $\svireach$ and $\svistay$ depend on a step-count $k$.
	Collapsing ECs results in not counting steps within ECs at all, while deflating performs multiple steps by jumping to the best exit. 
	In contrast, our handling of ECs needs to count steps exactly. 
	For this, we shall introduce a \enquote{delay} action in Sec.~\ref{ssec:new-notions}, which does not move us from the current state, yet counts as making a step.	
\end{enumerate}

We highlight that difficulties (\textbf{1}), (\textbf{3}), and (\textbf{4})
already arise in MDP, since the respective examples only contain Maximizer states.
Overall, deflating an entire SEC to its best exit harms correctness; and, deflating only the best exiting state $\bestExit_f^\Box$ leads to non-termination. 
As our solution goes in the direction of deflating exactly the best exiting states, we include an example for the arising complications here:


\begin{example}
	\label{ex:2stateCyclicDependency}
	
	Consider the SG in Fig.~\ref{fig:2stateCyclicDependency}. 
	It contains a MEC $T = \unknown = \{s_0,s_1\}$.
	Both these are Maximizer states and have value $\nicefrac{1}{2}$, achieved by playing at least one of the exiting actions infinitely often.
	Upon initialization, we have: 
	$\svireach^0_{s_0}=0,\svistay^0_{s_0}=1,\svireach^0_{s_1}=0,\svistay^0_{s_1}=1,\glb_0=0, \gub_0=1$.
	Comparing the exits from $T$, for $(s_0,b)$, the approximation $\svireach+\svistay\cdot\gub$ evaluates to $\nicefrac{1}{3}+\nicefrac{1}{3}\cdot1=\nicefrac{2}{3}$, which exceeds the approximation for $(s_1, b)$ given by $0.4+0.2\cdot1=0.6$.  
	Therefore, the best exit of $T$ in the first iteration is $\bestExit^\Box(T)=\{(s_0, b)\}$. 
	Forcing Maximizer to take this best exit leads to the following updated vectors:
	$\svireach^1_{s_0}=\nicefrac13,\svistay^1_{s_0}=\nicefrac13, \svireach^1_{s_1}=0, \svistay^1_{s_1}=1,
	\glb_1=0, \gub_1=1.$
\end{example}

	\begin{wrapfigure}{r}{0.5\textwidth}
			\centering
			\resizebox{0.5\textwidth}{!}{
				\begin{tikzpicture}
					\node[max vertex, ] (s0) {$s_0$};
					\node[max vertex, right of=s0] (s1) {$s_1$};
					\node[chance state, left of=s0, node distance=1cm](c0){};
					\node[chance state, right of=s1, node distance=1cm](c1){};
					\node[goal, above left of=s0, node distance=2cm] (t0) {};
					\node[goal, above right of=s1, node distance=2cm] (t1) {};
					\node[sink, below left of=s0, node distance=2cm] (z0) {};
					\node[sink, below right of=s1, node distance=2cm] (z1) {};
					
					\path (s0) edge[bend left=30] node[above] {$a$} (s1);
					\path (s1) edge[bend left=30] node[below] {$a$} (s0);
					\draw (t0) to [out=110,in=170,loop,looseness=7] (t0);
					\draw (t1) to [out=70,in=10,loop,looseness=7] (t1);
					\draw (z0) to [out=290,in=250,loop,looseness=8] (z0);
					\draw (z1) to [out=250,in=290,loop,looseness=8] (z1);
					\draw (s0) to [out=150,in=210,loop,looseness=5] node[left] {$b$} node[pos=0.15, anchor=north] {\scriptsize{\nicefrac{1}{3}}} (s0);
					\draw (s1) to [out=-30,in=30,loop,looseness=5]  node[right] {$b$} node[pos=0.15, anchor=south] {\scriptsize{$0.2$}}(s1);
					\path (c0) edge node[pos=0.6, anchor=east] {\scriptsize{\nicefrac{1}{3}}} (t0);
					\path (c0) edge node[pos=0.6, anchor=east] {\scriptsize{\nicefrac{1}{3}}} (z0);
					\path (c1) edge  node[pos=0.6, anchor=west] {\scriptsize{0.4}}(t1);
					\path (c1) edge  node[pos=0.6, anchor=west] {\scriptsize{0.4}} (z1);
				\end{tikzpicture}
			}
			\caption{SG with two maximizer nodes and self-loops}
			\label{fig:2stateCyclicDependency}
		\end{wrapfigure}
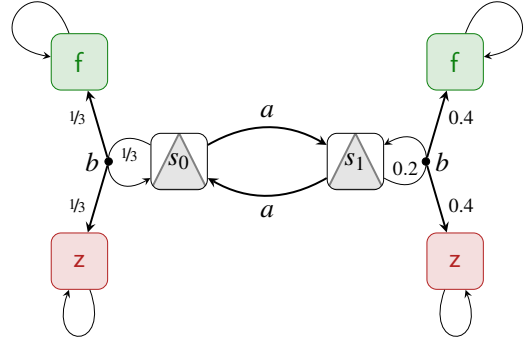
	In the second iteration, taking  $(s_0, b)$ again would decrease  $\svistay_2^{s_0}$ to $\nicefrac13\cdot\nicefrac13=\nicefrac19$ while increasing $\svireach_2^{s_0}$ to $\nicefrac13\cdot 1+\nicefrac13\cdot\nicefrac13=\nicefrac49$, yielding $\nicefrac49+\nicefrac19\cdot 1=\nicefrac59$.
	This is less than one step of $(s_1,b)$ yielding $0.4+0.2\cdot1=0.6$.
	Consequently, updating the best exit results in: 
	$\svireach^2_{s_0}=0,\svistay^2_{s_0}=1,\svireach^2_{s_1}=0.4, \svistay^2_{s_1}=0.2,\glb_2=0 ,\gub_2=1.$
In the third iteration, $\{(s_0, b)\}$ again is the best exit, which results in updating the vectors to the same values as in the first iteration.
This oscillation continues indefinitely.
As a result, the bounds remain unchanged, i.e.\ $\glb_i=0$, $\gub_i=1$  for all the subsequent iterations.

The non-termination arises because the algorithm keeps shifting the probability of remaining undecided entirely between two states (as reflected by $\svistay=1$ for either $s_0$ or $s_1$ in each iteration).  
To address this, we enforce monotonic progress of $\svireach+\svistay\cdot\gub$ (see Eq.~\eqref{eqn:impUB}); if updating a state would break monotonicity, it instead plays a \enquote{delay} action, see Sec.~\ref{ssec:new-notions}.

\subsection{Our solution}
We propose a new approach that updates the value vectors consisting of two key ideas: 1) a \enquote{recusive variant} of deflate which identifies the best exits by recursively cutting ECs further to break all cyclic dependency, and
2) a \enquote{delay condition} deciding whether to play a regular action or to wait in the current state by playing a \enquote{delay} action,
enforcing monotonicity and progress of upper bounds.
This method is also within the spirit of SVI, i.e.\ does not require to maintain local upper bound for every state.
Our solution can be seen as Alg.~\ref{alg:svi-sg-high-level}, where the computation of the strategies and bounds are adjusted by the re-defined procedures described in the following subsections.

\subsubsection{Best exit set (BES)}\label{ssec:bes}
This section examines the hierarchical structure of ECs and their respective exits. 
The procedure \textsc{BestExitSet} recursively identifies the best exits for each sub-EC, resulting in a comprehensive set of the different best exits---capturing relevant state-action pairs that serve as exits, each corresponding to some EC.
We remark that the set of best exits is empty for ECs that cannot be exited at all.

\begin{algorithm*}[!ht]
	\caption{Compute best exit set for the sound value iteration} 
	\label{alg:bestexitset-high-level}
	\begin{algorithmic}[1]
		\Require an SG $\G$, over-approximation $f\geq \val$, an EC $\mec$
		\Ensure Set of correct exiting state-action pairs of ECs in $\mec$ wrt $f$
		\Procedure{BestExitSet}{$\G, f, \mec$}
		\State BES $\gets \emptyset$ \Comment{Initialization}

		\If {$\bestExit^\Box_f(\mec)$ is empty} \label{line:BES-hl:BestExitEmpty}\Comment{Base case}
		\State Add the states of $\mec$ to $\sinks$ and remove them from $\unknown$ \label{line:BES-hl:removeTrapECs}
		\State \Return BES \label{line:BES-hl:returnBaseCase}

		\Else \Comment{Recursive step: {$\bestExit^\Box_f(\mec)$ is non-empty}}
		\State Add the best-exits of $\mec$ to BES \label{line:BES-hl:insertBE}
		\State Find MECs in the induced graph after removing all the states in BES\label{line:BES-hl:MECsInInducedGraph}
		\State For every such MEC $M$, add \textsc{BestExitSet}($\G, f, M$) to BES \label{line:BES-hl:callBESOnAllMECs}
		\EndIf
		\State \Return BES {\Comment{Set of correct exits \emph{wrt} $f$ for all EC within $\mec$}}\label{line:BES-hl:return2}
		\EndProcedure
	\end{algorithmic}
\end{algorithm*}

Alg.~\ref{alg:bestexitset-high-level} describes the procedure to compute the \emph{best exit set (BES)} for a given EC \(\mec\) relative to a function \(f \geq \val\).
An equivalent, fully formal description is in the extended version~\cite{svi-sg-extended}.
If there is no Maximizer exit (Lines~\ref{line:BES-hl:BestExitEmpty}-\ref{line:BES-hl:returnBaseCase}), then \phantomsection\newtarget{def:trap}{the Minimizer can force any play in the EC to remain there forever; we call such ECs a \emph{trap ECs}}
and denote the set of all such ECs as \trap. 
All states in $\trap$ can be included in the sink states $\sinks$ and removed from the unknown states set $\unknown$ (Line~\ref{line:BES-hl:removeTrapECs}).
Otherwise, all the best exit pairs are added to a set BES (Line~\ref{line:BES-hl:insertBE}). 
After removing all the states of BES, the procedure is recursively called on all the  MECs of induced graph (Lines~\ref{line:BES-hl:MECsInInducedGraph},~\ref{line:BES-hl:callBESOnAllMECs}). Finally, Line~\ref{line:BES-hl:return2} returns exits for all the relevant sub-ECs of $\mec$. 

The recursive nature of the procedure leverages the \emph{inductive, attractor-like structure of MECs}, providing a systematic decomposition that identifies all relevant exits while avoiding cyclic dependencies. 
Importantly, when the best exit of a MEC is identified removed, this often eliminates multiple sub-ECs simultaneously. 
This hierarchical dependency inherently reduces the effective number of ECs that must be considered in subsequent recursive calls: It is linear in the number of states in the MEC despite the potential exponential total number of ECs.
We prove
the following property of Alg.~\ref{alg:bestexitset-high-level}  in~\cite{svi-sg-extended}: It is sound insofar that every state-action pair included in the BES has a higher over-approximation than value; and for every sub-EC, some exit is included in the BES (or it has been merged into the sink states).

\begin{restatable}{lemma}{lemmaBestExitProps}
		\label{lem:bestExitSetProps}
	For every EC $T$ and $f \geq \val$, Alg.~\ref{alg:bestexitset-high-level} terminates and returns a set BES that has the following two properties:
	(I)~\emph{Soundness of over-approximation:} for all $(s,a) \in \text{BES}: f(s,a) \geq \val(s)$. 
	(II)~\emph{Completeness of BES:} For all ECs $T'\subseteq T$, if  $T \notin \trap$, then $\exists~ (s,a) \in \text{BES}: (s,a)$ exits $T$. Otherwise, if $T \in \trap$, then $\forall~ s \in T: s \in \sinks \land s \not\in \unknown$.
	($(s,a)$ exits $T$ denotes that $\post(s,a) \not\subseteq T$.)
	
\end{restatable}

At each iteration $k$, we can compute the overall set $\bestExitSet$ consisting of unions of all the best-exit sets for all ECs, based on the current over-approximation: 
\begin{equation}
	\boxed{
		\begin{aligned}
			\label{eqn:HandleECs}
			\bestExitSet^k \gets \bigcup_{\substack{\mec \subseteq \unknown\\ \mec\text{ is MEC}}}\textsc{BestExitSet}(\G, \svireach^k_s+\svistay^k_s \cdot \gub_k, \mec)
		\end{aligned}
	}
\end{equation}

\subsubsection{Delays and strategies}
\label{ssec:new-notions}
The \emph{delay} action helps to ensure monotonic progress. 
We refer the reader to  Ex.~\ref{ex:2stateCyclicDependency} 
 for an explanation of non-termination and cyclic behaviour without the delay action. To find out in which state we need to play a delay action, we first introduce a condition `\textsc{improvedUB}?($s$)' that prevents the Maximizer state to update the approximation to worse values than those of current estimates (monotonicity):

\begin{equation}
	\label{eqn:impUB}
	\boxed{
		\begin{aligned}
			\textsc{improvedUB}?(s) := \;
			 \svireach^{k+1}_s + \svistay^{k+1}_s \cdot \gub_k  \leq \svireach^k_s + \svistay^k_s \cdot \gub_k
		\end{aligned}
	}
\end{equation}

The following condition then captures when to delay rather than to play another action.
{\small
\begin{equation}
	\label{eq:delayaction}
	\boxed{
		\begin{aligned}
			\textsc{delayAction}?(s) := s \in \unknown_{\Box}  \land \lnot \textsc{ImprovedUB}?(s)
		\end{aligned}
	}
\end{equation}
}

Intuitively, if a Maximizer state $s$ is not a best exit of an EC and the upper bound for $s$ does not improve, we play a delay action which we denote as $d$. In particular, under the condition \textsc{delayAction}$(s)?$, the approximations for state $s$ at iteration $k+1$ are restored to the values from the $k$th iteration as detailed in the new part of Eq.~\eqref{alg:update-ec}:

\begin{equation}
	\fontsize{9.1}{12}\selectfont
	\label{alg:update-ec}
	\boxed{
		\begin{aligned}
			&\textsc{Bellman-Update ($\mathcal{B}$):}\\
			&\mbox{if~}{\textsc{delayAction?}(s)} & \triangleright \mbox{ New part to handle ECs}\\ 
			&\hspace{20pt}\svireach^{k+1}_s \gets \svireach^{k}_{s};~~~\svistay^{k+1}_s \gets \svistay^{k}_{s}\\
			& \mbox{else} & \triangleright \mbox{ Same as before}\\
			&\hspace{20pt}\svireach^{k+1}_s \gets \sum_{s' \in \states}\distribution(s,\pi_{k+1}(s), s') \cdot \svireach^{k}_{s'};~~~&\svistay^{k+1}_{s} \gets \sum_{s' \in \states}\distribution(s,,\pi_{k+1}(s), s') \cdot \svistay^{k}_{s'}\\			
		\end{aligned}
	}
\end{equation}


Finally, we re-define the strategies $\sigma_k$ with the new cases of exits and delays (while $\tau_k$ is unchanged):


\begin{equation}
	\fontsize{8.1}{11}\selectfont
	\label{eqn:newstrat}
	\boxed{
		\begin{aligned}
			&\textsc{``$k$-step'' Maximizer's contracting Strategies:}\\
			&\sigma_0(\rho) \in \Av(\rho_0) \\	
			&\sigma_{k+1}(\rho) \gets
			\begin{cases}
				\displaystyle
								\sigma_k(\rho') &\mbox{if } \rho = sa\rho'\\
								d &\mbox{if }  \rho= s \land \textsc{delayAction?}(s) \\
								b &\mbox{if }  \rho= s \land (s,b)\in \bestExitSet^k\\
								\displaystyle\in\argmax_{a\in\Av(s)} \sum_{s' \in S} \distribution(s, a, s') \cdot \Big(\probability^{\sigma_k, \tau_k}_{s'}  (\reach^{\leq k}\targets) + \probability^{\sigma_k, \tau_k}_{s'}  (\stay^{\leq k}\targets) \cdot \gub_k \Big) &\mbox{otherwise}
			\end{cases}
		\end{aligned}
	}
\end{equation}

\subsubsection{SVI algorithm for SG with ECs}

First, we prove correctness of the approximations and the strategies in the general case with ECs, considering the two additional ``cases''. 
The proof can be found in the extended version~\cite{svi-sg-extended}.

\begin{restatable}{theorem}{theoremSVIGenSGEC}
	\label{thm:svi-gen-sg-ec} 

Fix an SG $\G$ with its state space partitioned in $\targets$, $\sinks$ and $\unknown$. Let $\glb_k, \gub_k \in \Reals$ be bounds as defined by Eq.~\eqref{eq:def-bounds}, i.e. for all $k \geq 0$, $\glb_k \leq \val(s) \leq \gub_k$ for all $s \in \unknown$.
For  $\sigma_k \in \stratsMax$ and $\tau_k \in \stratsMin$ defined as in Eqs.~\eqref{eqn:newstrat} and~\eqref{eq:strats}, respectively,
the following inequality holds:	


\begin{align*}
	\probability_s^{\sigma_k, \tau_k} (\reach^{\leq k} \targets) 
	+ \probability_s^{\sigma_k, \tau_k} (\square^{\leq k} \unknown) \cdot \glb_k 
	\leq \val(s) \leq 
	\probability_s^{\sigma_k, \tau_k} (\reach^{\leq k} \targets) 
	+ \probability_s^{\sigma_k, \tau_k} (\square^{\leq k} \unknown) \cdot \gub_k
\end{align*}

\end{restatable}

\begin{remark}
	Introducing delay action changes the semantic of $\probability_s^{\sigma_k,\tau_k} (\reach^{\leq k} \targets) $ and $\probability_s^{\sigma_{k},\tau_{k}} (\square^{\leq k} \unknown)$ in a way that it is now the optimal reach and stay probability for an $m \leq k$, i.e., \emph{at most $k$}. We prove that this does not hurt the correctness.
\end{remark}

Now we show how to modify Alg.~\ref{alg:svi-sg-high-level} to obtain an algorithm that works for SG with ECs.
First, we instantiate the Maximizer strategy $\sigma$ by Eq.~\eqref{eqn:newstrat}.
Second, we use the new update \textsc{Bellman-update} as in Eq.~\eqref{alg:update-ec}.
Finally, we use the following condition \textsc{updateGlobalBounds?}:


{\small
\begin{equation}
	\boxed{
		\begin{aligned}
			\label{eqn:updateGlobalBounds}
			\textsc{updateGlobalBounds}? := \forall s \in \unknown. \; 
			 \Big( \svistay_s < 1  \land \lnot \textsc{delayAction?}(s) \Big)
		\end{aligned}
	}
\end{equation}
}

 Intuitively, whenever we play a delaying action, we conservatively do not improve the global bounds~$\glb$ and $\gub$.
A key property for proving termination is the following monotonicity of the over-approximation, proven in the extended version of this paper~\cite{svi-sg-extended}.
\begin{restatable}{lemma}{lemmamonotonicU}
	\label{lem:monotonicU}
Fix $\gub, \glb$ such that $\glb \leq \val(s) \leq \gub$. Alg.~\ref{alg:svi-sg-high-level} with the new sub-procedures from Sec.~\ref{sec:algorithm-with-ec} (Maximizer strategy~\eqref{eqn:newstrat}, new \textsc{Bellman-update}~\eqref{alg:update-ec} and \textsc{isUpdateGlobalBounds?}~\eqref{eqn:updateGlobalBounds}) computes a monotonically non-increasing sequence of $\probability_s^{\sigma_{k},\tau_{k}} (\reach^{\leq k} \targets) +
\probability_s^{\sigma_{k}, \tau_{k}} (\square^{\leq k} \unknown) \cdot \gub$,
i.e.\

\begin{align*}
	\probability_s^{\sigma_{k+1}, \tau_{k+1}} (\reach^{\leq {k+1}} \targets) 
	+ \probability_s^{\sigma_{k+1}, \tau_{k+1}} (\square^{\leq {k+1}} \unknown) \cdot \gub 
	\leq 
	\probability_s^{\sigma_k, \tau_k} (\reach^{\leq k} \targets) 
	+ \probability_s^{\sigma_k, \tau_k} (\square^{\leq k} \unknown) \cdot \gub
\end{align*}

\end{restatable}

\begin{theorem}
	Fix an SG $\G$ with its state space partitioned in $\targets$, $\sinks$ and $\unknown$ and precision $\varepsilon>0$. 
	Alg.~\ref{alg:svi-sg-high-level} with the new sub-procedures from Sec.~\ref{sec:algorithm-with-ec} (Maximizer strategy~\eqref{eqn:newstrat}, new \textsc{Bellman-update}~\eqref{alg:update-ec} and \textsc{isUpdateGlobalBounds?}~\eqref{eqn:updateGlobalBounds}) is correct and terminates with an $\varepsilon$-close approximation to the value $\val'$, i.e. for all $s \in \unknown, |\val'(s)- \val(s)| \leq \varepsilon$.
\end{theorem}
\emph{Proof Sketch.} 
\para{Correctness.} For correctness, we argue that $\svireach_s^k$ and $\svistay_s^k$ are computing $\probability^{\sigma_{k}, \tau_{k}}_{s}(\reach^{\leq k}\targets)$ and $\probability^{\sigma_{k}, \tau_{k}}_{s}(\stay^{\leq k}\unknown)$ 
which is aligned with Thm.~\ref{thm:svi-gen-sg-ec}. We ensure that $\glb_k$ and $\gub_k$ are updated correctly in each iteration. 
Then, correctness follows from Thm.~\ref{thm:svi-gen-sg-ec}.

\smallskip \para{Termination.}
We sketch a contradiction-based proof of algorithm termination, even in the presence of ECs. The central hypothesis we refute is that the algorithm does not terminate, despite having reached the fixpoint of the over approximation. Here, we present a high level idea of the proof in five steps. For a detailed proof, please refer to the extended version~\cite{svi-sg-extended}.
\begin{enumerate}
    \item Firstly, using monotonicity of over-approximation and the property that $\gub_{k+1} \leq \gub_{k}$, we show that the function $\ub_k$, defined as $\forall s \in \states: \ub_k(s) \eqdef \svireach^k_s + \svistay^k_s \cdot \gub_k$  is well defined and has a fixpoint, i.e. $ \ub^*:=\lim_{k \to \infty}\ub_k$. 
	\item Secondly, we assume for contradiction that Alg.~\ref{alg:svi-sg-high-level} does not terminate even when the fixpoint in Step~1 is reached. 
	\item Using this assumption, we define a subsystem $X$ of $\G$, where the difference between $\ub^*$ and $\val$ is maximum, and show the existence of an EC within $X$.
	\item Next, we show that applying another iteration of Alg.~\ref{alg:svi-sg-high-level} leads to a contradiction to Step~2.
    In particular, we refute the assumption that $\ub^*>\val$.
	Thus, we can conclude $\ub^*=\val$. Details appear in~\cite{svi-sg-extended}.

	\item Finally, using $\ub^*=\val$ and a similar contradiction argument, we prove that in the limit $\forall s \in \unknown.~ \svistay^k_s=0$. Therefore, the function $\lb_k$, defined as $\forall s \in \states: \lb_k(s) = \svireach^k_s + \svistay^k_s \cdot \glb_k$  equals to $\svireach^k_s$, which implies that under and over-approximations converges to the same fixpoint.


\end{enumerate}

 \section{Experimental evaluation and illustrative example of SVI's efficiency}
  \label{svi-good-example}
 \myspaceb
 \begin{wrapfigure}{r}{0.4\textwidth}
 	\vspace*{-0.3in}
 	\centering\myspace\myspace\myspace\myspace
 	\resizebox{0.4\textwidth}{!}{
 	\begin{tikzpicture}
 		\node[max vertex] (s) {$q$};
 		\node[chance state, right of=s, node distance=1cm](c1){};
 		\node[goal, right of=c1, node distance=1.5cm] (t1) {};
 		\node[sink, below of=t1, node distance=2cm] (z1) {};
 		\node[min vertex, below of=s, node distance=2cm] (s') {$r$};
 		\node[initial, min vertex, below left of=s, node distance=1.50cm] (s0) {$p$};
 		\draw (t1) to [out=115,in=65,loop,looseness=6] (t1);
 		\draw (z1) to [out=245,in=290,loop,looseness=6] (z1);
 		\draw (s) to [out=25,in=-25,loop,looseness=6] node[left] {\scriptsize{$a$}}  node[pos=0.15, anchor=south]{\scriptsize{$0.98$}} (s);
 		\path (c1) edge[bend left] node[below] {\scriptsize{$0.01$}}(t1);

 		\path (c1) edge[bend left] node[left] {\scriptsize{$0.01$}}(z1);
 		\path (s0) edge[bend left] node[above] {\scriptsize{$a$}}(s);
 		\path (s0) edge[bend right] node[below] {\scriptsize{$b$}}(s');
 		\path (s') edge[bend right] node[above] {\scriptsize{$c$}}(t1);
 		\path (s) edge[bend right] node[above ] {\scriptsize{$b$}}(z1);
 	\end{tikzpicture}
 }
 	\caption{An SG with $\unknown=\{p,q,r\}$}
 	\label{fig:SG-SVI-good}	
 	\myspace\myspace

 \end{wrapfigure}
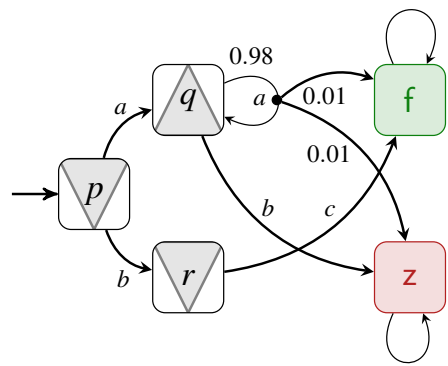
 Our prototype implementation of SVI demonstrates comparable iteration efficiency to BVI overall, while outperforming it on benchmarks where traditional VI shows slow convergence issues (due to probabilistic cycles). 
 The extended version~\cite{svi-sg-extended} includes experimental comparisons between SVI and BVI.
 Below, we provide an example to showcase that the main advantage of SVI is preserved even in SG.
 \begin{example}[SVI solves SG faster]
 In Fig.~\ref{fig:SG-SVI-good}, for the Minimizer in state $p$, playing $a$ is optimal as $b$ would give a value of 1. In state $q$, for the Maximizer playing $a$ is better because $b$ would give value of $0$. 
 
\noindent To solve this small example, BVI needs 685 iterations.
Our new version of SVI is now applicable, even though the example is an SG. 
It solves it in just 2 iterations because it can deal with the probabilistic loop on state $q$ in just one iteration.
 \end{example}

	\section{Topological variants of sound value iteration}
\label{sec:topological}
The core idea of topological value iteration~\cite{TVI1} is to exploit the fact that there is an ordering among states, and certain states can never be reached again. 
More formally, a transition system can be decomposed into a directed acyclic graph of strongly connected components (SCC); the topological variant of a solution algorithm proceeds backwards through this graph, solving it component by component, see~\cite[Sec.~4.4]{MaxiGandalf-journal} for more details.
Also, we note that using the standard topological approach can in practice lead to non-termination~\cite[Sec.~4]{atva22ovi-tvi}.

We propose a new approach of exploiting the topology of the transition system, specialized for SVI.
The strength of SVI is in picking good bounds $\glb$ and $\gub$ quickly. However, the bounds are picked as the minimum/maximum over all states in $\unknown$; thus, as long as there is one state that does not yet have a lot of information, it worsens the bounds for all others. 
In particular, if there is one state with a staying probability of 1, the bounds cannot be updated at all.
These bad bounds additionally slow down convergence because the choice of strategies depends on the bounds.

To address this problem, we propose to modify the definition of the bounds, in particular the term $\opt_{s \in \unknown}~ \frac{\svireach_s^k}{1-\svistay_s^k}$ in Eq~\eqref{eq:def-bounds}, as follows:
instead of globally defining the bounds by picking the optimum among all unknown states, we define a bound for each state $s$ which only picks among those states that are reachable from $s$.\footnote{Formally, let $\mathsf{All\_reachable\_states}(s)$ be the set of states reachable from $s$. Then the bounds for $s$ use the term $\opt_{s \in \unknown \cap \mathsf{All\_reachable\_states}(s)}~ \frac{\svireach_s^k}{1-\svistay_s^k}$}
This is correct because the bound represents the lowest/highest possible value the state $s$ can achieve after staying for $k$ steps; and, naturally, from $s$, the play can only reach a state that is reachable from $s$.
In practice, we do not need to store a bound for every state, but instead one per SCC, because all states in the same SCC can reach each other, and thus also share their set of reachable states.

This improvement allows for more informed choices in parts of the state space where $\svireach$ and $\svistay$ already suffice to derive good bounds, without being hampered by other slow or uninformed parts of the state space. Moreover, we can stop updating an SCC when its lower and upper bound are equal, saving some resources.

	\section{Conclusion}
We have extended the sound value iteration from Markov decision processes to stochastic games.
In order to achieve that we had to lift the key assumption requiring that there are no end components. 
While the literature already suggests approaches to do so for other variants of value iteration, we demonstrated that they do not apply to sound value iteration.
Consequently, we had to design a new dedicated solution.
Additionally, we have proposed improvements exploiting the topological properties of the models.
While the approach is mostly on par with other value iteration approaches providing guarantees on precision, it shows a clear advantage in systems with probabilistic cycles.
Yet, the point here was not a faster tool, but understanding the structure, in particular principles of evaluating cycles in SG. 
The inductive, attractor-like structure of MECs provides a key theoretical foundation for handling cycles. Furthermore, it addresses a critical drawback of VI and opens the path for potential future enhancements to BVI through efficient probabilistic cycle handling.
Further, besides optimized implementation, the future work could include the theoretically interesting extension to concurrent stochastic games.
	\bibliographystyle{eptcs}
	\bibliography{ref}
\end{document}